\PassOptionsToPackage{unicode}{hyperref}
\PassOptionsToPackage{hyphens}{url}
\documentclass[10pt]{article}
\usepackage{amsmath,amssymb,amsthm}
\usepackage{lmodern}
\usepackage{iftex}
\usepackage{graphicx}
\ifPDFTeX
  \usepackage[T1]{fontenc}
  \usepackage[utf8]{inputenc}
  \usepackage{textcomp} % provide euro and other symbols
\else % if luatex or xetex
  \usepackage{unicode-math}
  \defaultfontfeatures{Scale=MatchLowercase}
  \defaultfontfeatures[\rmfamily]{Ligatures=TeX,Scale=1}
\fi

\usepackage[square,numbers]{natbib}

\IfFileExists{upquote.sty}{\usepackage{upquote}}{}
\IfFileExists{microtype.sty}{% use microtype if available
  \usepackage[]{microtype}
  \UseMicrotypeSet[protrusion]{basicmath} % disable protrusion for tt fonts
}{}
\makeatletter
\@ifundefined{KOMAClassName}{% if non-KOMA class
  \IfFileExists{parskip.sty}{%
    \usepackage{parskip}
  }{% else
    \setlength{\parindent}{0pt}
    \setlength{\parskip}{6pt plus 2pt minus 1pt}}
}{% if KOMA class
  \KOMAoptions{parskip=half}}
\makeatother
\usepackage{xcolor}
\IfFileExists{xurl.sty}{\usepackage{xurl}}{} % add URL line breaks if available
\IfFileExists{bookmark.sty}{\usepackage{bookmark}}{\usepackage{hyperref}}
\hypersetup{
  hidelinks,
  pdfcreator={LaTeX via pandoc}}
\ifLuaTeX
  \usepackage{selnolig}  % disable illegal ligatures
\fi
\usepackage{algorithm}
\usepackage{algorithmic}

\newtheorem{definition}{Definition}
\newtheorem{lemma}{Lemma}
\newtheorem{theorem}{Theorem}
\newtheorem{proposition}{Proposition}

\title{Stackelberg Attack on Protocol Fee Governance}
\author{Alexandre Hajjar\thanks{lajarre@buttery.money} \\ Buttery Good Games}
\date{May 2024}

\begin{document}
\maketitle
\setcounter{tocdepth}{2}

\setcitestyle{numbers}

\let\oldabstract\abstract
\let\oldendabstract\endabstract
\renewenvironment{abstract}
  {\renewcommand{\baselinestretch}{1}\oldabstract\noindent\ignorespaces}
  {\oldendabstract}
\begin{abstract}
We establish a Stackelberg attack by Liquidity Providers against Governance of
an AMM, leveraging forking and commitments through a Grim Forker smart contract.
We produce a dynamic, block-by-block model of AMM reserves and trading volume in
the presence of competing forks, derive equilibrium conditions in the presence
of protocol fees, and analyze Stackelberg equilibria with smart contract moves.
\end{abstract}

\section{Introduction}

\subsection{Related work}

Protocol fees are used by various DeFi protocols to allow value capture by their
governance smart contracts \cite{UniswapV3, morphoblue}. Fees can then be used
to produce dividends for token holders or to accumulate treasury value
\cite{gfxuniswap} thus increasing the governance token's value and justifying
investability in the protocol. Notably, major protocols like Uniswap haven't
activated the protocol fee yet but are working towards it
\cite{UniswapFeeProposal}.

Protocol fees are taken from some protocol participants' revenue. In Uniswap's
case, protocol fees are deducted from Liquidity Provider (LP) fees
\cite{UniswapV3}. This dynamic creates an adversarial setting where LPs have
diverging interests from token holders. Given this occurs in a blockchain
setting, it is legitimate to ask if and how LPs can leverage smart contracts to
modify the strategic equilibrium to their advantage.

Various models in the literature account for a player whose strategy leverages
commitments to alter the game's outcome. \textit{Stackelberg competition games}
introduced in \cite{VonStackelberg} model competitive games with a leader and a
follower, where the leader commits to a move. \cite{CommitmentFolk} formalizes
\textit{commitment games}, where players can commit to playing a fixed strategy.
\textit{Game mining}, introduced by \cite{GameMining}, describes commitment
games whose outcomes can be altered by one player committing publicly to produce
conditional side-payments. 

Blockchains introduce a new setting where non-cooperative players can achieve
commitments, thus producing novel game equilibria. \cite{VirgilGriffith} uses
the term \textit{game warping} to describe how non-cooperative games' equilibria
are modified by players interacting with smart contracts.
\cite{Ramirez2023GameM} studies how \textit{commitment devices}, including
blockchain-enabled smart contracts as binding contracts, can be used by players,
particularly \textit{game manipulators} external to the original game, to alter
payoffs in their favor. \cite{GThBlockchain} models Stackelberg equilibria where
smart contract moves are introduced in a game.

Smart contract commitments have a history of being leveraged to create
theoretical attacks in the blockchain setting, like in
\cite{cryptoeprint:2017/230}. More recently, literature has developed around
\textit{Stackelberg attacks} using smart contracts as commitment devices to
induce new Stackelberg equilibria to the attacker's profit \cite{Huo_2023}
\cite{landis2023stackelberg}. \cite{Landis2023WhichGA} introduces
\textit{Stackelberg resilience} as a property of games whose subgame perfect
equilibria are not changed by sequential commitments.

Stackelberg attacks have been studied only sparsely in the context of smart
contract governance. \cite{Huo_2023} produces conditions on Stackelberg
equilibria between stakeholders of a stablecoin system, including governance as
a player and smart contract parameter setting as the action space. Such models
are yet to be studied in the context of AMM protocol fee governance.

\subsection{This paper}

In this paper, we aim to produce a framework for deriving Stackelberg equilibria
between, on one hand, LPs allocating reserves among an AMM and a fork of it, and
on the other hand, the AMM Governance as a single player.

The analysis will be limited to a single Constant Product Market Makers (CPMM)
pool like Uniswap v2, of which we briefly introduce the core components and
refer to \cite{angeris2021analysis} for further details. However, note that the
results can be generalized to multiple pools and to Uniswap v3 tick-by-tick.
While the model focuses on CPMMs, it can be extended to a wider class of AMMs
whose trader and LPs allocation functions display similar properties.

The second section will formalize a game-theoretic model of competition among
two forks of the same AMM pool, where we observe the rational behavior of AMM
users (LPs, traders, and Governance). The third section will then study these
competition games in a dynamic system setting, block-by-block, with protocol fee
rate and initial conditions as parameters. This will allow observing a natural
monopoly where the leading AMM tends to win the whole market, reflecting the
idea that network effects apply.

Finally, in the fourth section, we introduce a meta-game played by LPs against
Governance, where they can commit to strategies that automatically deploy their
reserves among an AMM and a feeless fork of it, depending on conditions. The
intuition is that by colluding, LPs will threaten to reallocate most of the
liquidity to the feeless fork, thus gaining enough bargaining power to deter
Governance from increasing the protocol fee. We introduce the Grim Forker
contract as an instance of a Stackelberg attack by LPs to that effect. This
contract produces a new Stackelberg equilibrium to the competition game,
allowing us to define which protocol-fee-setting strategies Governance should
rationally follow.

\section{Model}

\subsection{Automated Market Maker}

We define an Automated Market Maker as being a 2-token pool, similar to Uniswap
v2.

\begin{definition}[Automated Market Maker (AMM)]
An Automated Market Maker (AMM) is defined by the tuple $(\alpha, \beta, R, V,
\gamma, \phi)$, where:
\begin{itemize}
    \item $\alpha$ and $\beta$ are the token types,
    \item $R$ represents the amounts in reserves,
    \item $V$ is the traded volume over a period of time,
    \item $1 - \gamma$ is the percentage fee,
    \item $\phi$ is the protocol fee.
\end{itemize}
\end{definition}

By convention, $R$ denotes the reserves in $\alpha$, and $V$ denotes the volume
in $\alpha$. For the rest of this paper, any quantity that denotes an asset will
use the same convention.

In the following allocation games, we will consider two AMMs which are in the
most direct kind of competition, or \textit{Competing AMM Forks}.

\begin{definition}[Competing AMM Forks]
  Two AMMs $(\alpha, \beta, R_a, V_a, \gamma, \phi_a)$ and $(\alpha, \beta, R_a,
  V_a, \gamma, \phi_a)$ are competing when:
  \begin{itemize}
    \item they share the same token types $\alpha$ and $\beta$,
    \item they share the same fee level $\gamma$,
    \item they are operated by the same smart contract code on the same
      blockchain, thus resulting in identical blockchain transaction costs for
      performing a trade.
  \end{itemize}
\end{definition}

In the rest of this paper, when considering Competing AMM Forks, $\mathcal{A}_x$ will
systematically denote $(\alpha, \beta, R_x, V_x, \gamma, \phi_x)$.

For the purposes of analyzing equilibria under competition, we will want to
observe two competing AMM forks where one is a clear market leader with the
following condition.

\begin{definition}[Market Leader Condition]
  Two Competing AMM Forks $\mathcal{A}_a$ and $\mathcal{A}_b$ verify the Market
  Leader Condition when $V_a > V_b$ and $R_a > R_b$.
\end{definition}

This will also enable our analysis to focus on cases where the ratio of reserves
always starts strictly higher than $50\%$ in favor of the larger AMM.

\subsection{Trader volume allocation}

We now model traders and LPs so as to reason about how they will behave when
presented with the option to allocate their trades or reserves to any of two
competing AMM forks.

We assume traders to be profit maximizing. Their net value received from a trade
on a given AMM is modeled as $\textsf{TradeValue} - \textsf{Fees} -
\textsf{PriceImpact} - \textsf{TxCost}$, with:
\begin{itemize}
    \item $\textsf{TradeValue}$: the payoff from trading, specific to the
      trader's preferences and independent of the AMM,
    \item $\textsf{Fees}$: the AMM trading fees,
    \item $\textsf{PriceImpact}$: the cost resulting from the price impact of
      the trade, specific to the AMM reserves and to the size of the trade,
    \item $\textsf{TxCost}$: the blockchain transaction costs, specific to the
      AMM, but equal when comparing Competing AMM Forks.
\end{itemize}

The $\textsf{PriceImpact}$ term can be derived from \cite{angeris2021analysis}
equation (7) where we consider the price impact to be equal to the price gap
compared to a perfect market with infinite reserves. For a given AMM
$\mathcal{A}_a$:

$$
\textsf{PriceImpact} = m_u\gamma^{-1}\left(\frac{(\delta\Delta)^2}{R_a} + O\left(\frac{(\delta\Delta)^2}{R_a^2}\right) \right)
$$

with $m_u$ the AMM price of coin $\alpha$ and $\Delta$ the amount of coin
$\alpha$ traded.

The $\textsf{TxCost}$ term plays an important role in that it contributes an
incentive for traders to use a single AMM for their trades rather than spread
them freely over multiple AMMs. This incentive will impact the equilibrium
condition and result in network effects appearing.

This term assumes the traders will pay the AMM tx costs for each interaction
independently. Other options would exist like using Dex aggregators which would
produce a different cost profile, but it is reasonable to assume any such cost
profile would always result in an incentive to use a single smart contract
interaction rather than multiple ones.

\begin{definition}[Trader Allocation Game]
  A trader aims to exchange tokens $\alpha$ for tokens $\beta$, with the choice to
  allocate trading volume $\Delta$ among two competing AMM forks $\mathcal{A}_a$ and
  $\mathcal{A}_b$. Let $\delta \in [0, 1]$ represent the proportion allocated to
  $\mathcal{A}_a$, with the remaining fraction $(1-\delta)$ allocated to
  $\mathcal{A}_b$.

  The trader's utility is defined as

  \[
  \begin{gathered}
    U_t(\Delta, \delta) =
    \textsf{TradeValue}(\Delta) \\
    - \textsf{Fees}(\Delta)   \\
    - m_u\gamma^{-1}\left(\frac{(\delta\Delta)^2}{R_a} + \frac{((1-\delta)\Delta)^2}{R_b} + O\left(\frac{(\delta\Delta)^2}{R_a^2}\right) + O\left(\frac{((1-\delta)\Delta)^2}{R_b^2}\right) \right) \\
    - 
    \begin{cases}
      c_0 & \text{if } \delta \in \{0,1\} \\
      2c_0 & \text{otherwise}
    \end{cases}
  \end{gathered}
  \]
  where $c_0$ the transaction cost for performing a swap on either of the two
  AMMs, considered constant, and $m_u$ the ratio between reserves in token $\beta$
  and token $\alpha$, considered equal across AMMs (no-arbitrage condition).
\end{definition}

Note that the blockchain transaction costs are doubled if the allocation is
split.

As we are only looking for equilibria in an allocation game, let's use a
simplified utility keeping only terms that vary with the allocation proportions
and simplifying out fixed multiplicative terms. Let's also consider only small
enough trades so that quadratic terms are negligible.

\begin{definition}[Small Trader Allocation Simplified Utility]
Assuming two competing AMM forks, we define the trader's simplified utility,
assuming $\Delta$ fixed, as

  \begin{equation}
u_t(\delta) = - \frac{\delta^2}{R_a} - \frac{(1-\delta)^2}{R_b} - 
    \begin{cases}
      c & \text{if } \delta \in \{0,1\} \\
      2c & \text{otherwise}
    \end{cases}
  \end{equation}
with $c = c_0 \gamma m_u^{-1} \Delta^{-2}$.
\end{definition}

This is a concave function on $(0, 1)$ discontinuity at $\delta = 0$ and $\delta = 1$.
We can easily see that the maximum on $(0,1)$ will be the proportions of
reserves.

\begin{lemma}
  $$
  \underset{(0,1)}{\arg\max}\ u_t = \frac{R_a}{R_a + R_b}
  $$
\end{lemma}

Thus we observe that under the Market Leader Condition, some traders will prefer
to allocate all their trade volume to the leading AMM rather than seek 
the maximum described above, as long as the blockchain transaction costs are
large enough to them. This will be further studied below in our equilibria
analysis.

Let's denote $\sigma$, the proportion of such traders. We will assume that this
proportion is constant throughout Trader Allocation Games under the Market
Leader Condition in this paper. This reasonable assumption will be enough to
enable network effects.

\begin{definition}[Sensitive Traders Proportion]
  In any Trader Allocation Game under the Market Leader Condition, $\sigma$ in
  $[0,1]$ is defined so that:
  \begin{itemize}
    \item $\sigma$ is the proportion of traders for whom
      $\underset{[0,1]}{\arg\max}\ u_t = 1$,
    \item $(1 - \sigma)$ is the proportion of traders for whom
      $\underset{[0,1]}{\arg\max}\ u_t = \frac{R_a}{R_a + R_b}$.
  \end{itemize}
\end{definition}

\subsection{LP reserves allocation}

\begin{definition}[LP Allocation Game]
  A Liquidity Provider (LP) aims to allocate reserves (assumed only in token
  $\alpha$ for simplicity), with the option to allocate a total of $r$ among
  two competing AMM forks $\mathcal{A}_a$ and $\mathcal{A}_b$. Let $r_a \in [0,
  r]$ represent the amount allocated to $\mathcal{A}_a$, with the
  remaining $r_b = r - r_a$ allocated to $\mathcal{A}_b$.

  The LP's utility is given by 

  \[
    U_l(r, r_a) = (1-\gamma-\phi_a)V_a\frac{r_a}{R_a+r_a} +
    (1-\gamma-\phi_b)V_b\frac{r-r_a}{R_b+r-r_a} 
  \]
\end{definition}

$U_l$ is based on the definition of LP returns and protocol fee in
\cite{UniswapV3}.

This results in the Taylor expansion:

\[
    U_l(r, r_a) = (1-\gamma-\phi_a)V_a\frac{r_a}{R_a} +
  (1-\gamma-\phi_b)V_b\frac{r-r_a}{R_b} + O\left(\frac{r_a^2}{R_a^2}\right) + O\left(\frac{(r-r_a)^2}{R_b^2}\right)
\]

Let's thus approximate the utility for small LPs, who have reserves negligible
compared to the total reserves of any of the two AMMs.

\begin{definition}[Small LP Utility]
  \begin{equation}
  u_l(r, r_a) = (1-\gamma-\phi_a)V_a\frac{r_a}{R_a} + (1-\gamma-\phi_b)V_b\frac{r-r_a}{R_b}
\end{equation}
\end{definition}

Leveraging $u_t$ and $u_l$ will allow us to derive equilibria for both
allocation games.

\subsection{Aggregate allocation}

We want to further study the evolution through time of competing AMM forks
$\mathcal{A}_a$ and $\mathcal{A}_b$. For that, we want to observe the evolution
of $V_a/V$ and $R_a/R$.

This will enable analyzing which parameters influence equilibria and which
conditions produce (resp. prevent) a self-reinforcing network effect.

To that end, we define two sequences that capture the dynamic of repeated games
through blocks and model the occurrence of trades and reserves allocation as a
series of sequential allocation games within a block

We further want to assume that the aggregate volume and reserve allocated by
traders at any time step are constant, so as to focus our inquiry on the
dynamics between the two AMMs.

\begin{definition}[Block Allocation Game]
  Blocks are denoted by their indexes $i \in \mathbb{N}$.
  
  We consider $\mathcal{A}_a$ and $\mathcal{A}_b$ which verify:

  \begin{itemize}
    \item the total allocated trade volume per block, $V = (V_a)_i + (V_b)_i$ is
      constant for all $i\geq0$,
    \item the total amount of reserves $R = (R_a)_i + (R_b)_i$ is constant, for
      all $i\geq0$,
  \end{itemize}

  with AMM notation expanded to include $((V_{\cdot})_i)_{i\geq0}$ the trade
  volume per block and $((R_{\cdot})_i)_{i\geq0}$ the reserves amount per block.

  A Block Allocation Game is constructed by sub-games for each blcok $i\geq0$ in
  the following order:

  \begin{itemize}
    \item Block Traders Allocation Game: repeated sequential Trader Allocation
      Games so that the sum of trades allocated is equal to the volume per block
      $V$,
    \item Block LPs Allocation Game: repeated sequential LP Allocation Game so
      that the sum of reserves allocated is equal to $R$.
  \end{itemize}
\end{definition}

From here onward, any mention of AMM forks will refer to this definition.

Note that the entirety of reserves, $R$, is allocated anew by LPs at each block.
This simplifying measure will enable us focusing on how equilibria emerge, even
if sacrificing modeling accuracy.

Taking into account multiple forms of transaction costs for LPs would result in
only a fraction of $R$ being reallocated at each round, which makes the dynamic
system move slower, but arguably still in the same direction.

More generally, picking a different definition for Block Allocation Games might
yield slightly different results down the line, thus might be worth exploring to
refine the model.

To further simplify, we also assume that all traders playing the Block Traders
Allocation Games make small-enough trades. Hence, for the purpose of analyzing
equilibria, we will define their utility as the Small Trader Allocation
Simplified Utility $u_t$. Equivalently, we assume that all LPs entering any
Block LPs Allocation Game have reserves negligible compared to the total
reserves so we can define their utility as the Small LP Utility $u_l$.

Note that these simplifications will only matter as a way to calculate
equilibria, thus need only be good local approximations. For example,
if we assume that in an average Ethereum block a large proportion of LP reserves
allocation are of the small kind, these will be robust.

Nevertheless, we might be losing modeling accuracy with respect to blocks where
a single large LP is making a large move. Such events could produce sudden
variations, which could be modeled as stochastic terms. It is outside of the
scope of this simple model and could be suggested as a refinement.

Next, we define our dynamic system based on allocation ratios rather than volume
and reserves amounts, leveraging the Block-Competing AMM Forks constraints on $V$ and
$R$ being constant thus unnecessary to include in our model.

\begin{definition}[Allocation Ratios]
Given two Block-Competing AMM forks $\mathcal{A}_a$ and $\mathcal{A}_b$, along
with an initial block indexed $0$ (by convention), $(T_i)_{i\geq0}$ and
$(L_i)_{i\geq0}$ are sequences with values in $[0,1]$ defined by the recurrence relations:

\begin{itemize}
  \item $T_0 = \frac{(V_a)_0}{(V_a)_0 + (V_b)_0}$
  \item $L_0 = \frac{(R_a)_0}{(R_a)_0 + (R_b)_0}$
  \item $T_i = \textsf{BlockTradersAllocation}(T_{i-1}, L_{i-1})$
  \item $L_i = \textsf{BlockLPsAllocation}(T_i, L_{i-1})$
\end{itemize}
where
\begin{itemize}
  \item $\textsf{BlockTradersAllocation}$ represents the outcome of the Block
    Traders Allocation Game as a function of previous block state
  \item $\textsf{BlockLPsAllocation}$ represents the outcome of the Block LPs
    Allocation Game as a function of the outcome of Block Traders Allocation
    Game and of previous block state.
\end{itemize}
\end{definition}

$(T_i)$ represents the aggregate swap volume allocation ratio to
$\mathcal{A}_a$. $(L_i)$ represents the aggregate the aggregate reserves
allocation ratio to $\mathcal{A}_a$.

\section{Equilibria analysis}

\subsection{No fee scenario: network effects}

If two AMM forks follow the Market Leader Condition, we want to
observe how network effects apply and make the leader a monopoly. For that, we
analyze the dynamic system described by $(L_i)$ and $(T_i)$.

We will assume that $(L_i)>0.5$ and $(T_i)>0.5$ throughout this section unless
otherwise specified.

\begin{proposition}[Block Traders Allocation Rule]
  In a Block Allocation Game under the Market Leader Condition,
  $$
  \forall i > 0, T_i = \sigma + (1-\sigma)L_{i-1}
  $$
\end{proposition}
\begin{proof}
  Follows from the Sensitive Traders Proportion definition applied to every
  Trader Allocation Game in any Block Traders Allocation Games.
\end{proof}

\begin{proposition}
  In a Block Allocation Game under the Market Leader Condition with $\phi_a =
  \phi_b = 0$
  $$
  \forall i>0, L_i = T_i
  $$
\end{proposition}
\begin{proof}
  Follows from observing that the equilibrium condition for any LP Allocation
  Game in a Block LPs Allocation Game is $\frac{\partial u_l}{\partial r_a} = 0$
  which is equivalent to $\frac{R_a}{V_a} = \frac{R_b}{V_b}$ hence $\frac{R_a}{R} =
  \frac{V_a}{V}$.
\end{proof}

\begin{proposition}
  In a Block Allocation Game under the Market Leader Condition with $\phi_a =
  \phi_b = 0$
  $$
  \lim_{i\to\infty} L_i = 1
  $$
\end{proposition}
\begin{proof}
  $L_i = \sigma + (1-\sigma)L_{i-1}$ thus $(L_i)$ is strictly increasing. And it
  is bounded by 1.
\end{proof}

\subsection{Leader with fee scenario: equilibrium condition on fee}

If the market leading AMM enables a protocol fee $\phi_a > 0$ while the fork
doesn't, we expect that some players will allocate less to the leading AMM. We
want to look into when it produces a non-monopoly equilibrium.

In this sub-section, we assume $\mathcal{A}_a$'s protocol fee is permanently set
to $\phi_a$.

\begin{proposition}
  In a Block Allocation Game under the Market Leader Condition with
  $\phi_a > 0$ and $\phi_b = 0$

  $$
  \forall i>0, L_i = \frac{1-\gamma-\phi_a}{1-\gamma-\phi_a T_i}T_i
  $$
\end{proposition}
\begin{proof}
  Follows from deriving the maximum $u_l$ as a function of $r_a$, shown in
  equation (2), with $\phi_b=0$.
\end{proof}

We can observe that, as $T_i \in [0,1]$, for $i>0$ we have $L_i < T_i$. This
corroborates the idea that introducing the fee in $\mathcal{A}_a$ reduces its
attractiveness for LPs and thus hinders the network effects.

\begin{proposition}[Leader Fee Equilibrium]
  In a Block Allocation Game under the Market Leader Condition with
  $\phi_a > 0$ and $\phi_b = 0$, and assuming that $T_0$ is not null, the
  dynamic system will be at equilibrium either

  \begin{itemize}
    \item when $T_i = 1$,
    \item when $\phi_a = \sigma(1 - \gamma)T_i^{-1}$.
  \end{itemize}
\end{proposition}
\begin{proof}
  $T_0 \neq 0$ induces $T_i \neq 0$ for all $i$. Then the equation follows
  from solving the equilibrium equation constructed by combining the Leader
  Fee Block LPs Allocation Rule with the Block Traders Allocation Rule.
\end{proof}

Note that whenever $T_{i_0}$ = 1, it results that $T_i = S_i = 1$ for any $i >
i_0$ and the system won't move. This notably means that introducing any fee at
this point will not make the system move at all.

Typical values are $\gamma=0.003$ and $\phi=0.0006$ as per \cite{UniswapV3}. We
assume $\sigma=0.2$ to be a reasonable approximation. In these conditions,
$\sigma(1-\gamma)\phi_a^{-1} > 1$ thus making the second equilibrium condition
unachievable. In such cases, equilibrium can only be achieved when $T_i = 1$.

To illustrate aggregate players behavior depending on how the actual protocol fee
is set, let's suppose it is set so that $\phi_a < \sigma(1-\gamma)T_i^{-1}$.
This produces additional interest from LPs to allocate to $\mathcal{A}_a$, thus
making both reserves and volume allocated to $\mathcal{A}_a$ increase through
time until equilibrium is reached.

Also, whenever the fee is high enough, trade volume and reserves allocation to
$\mathcal{A}_a$ will eventually become lower than half of the total, canceling
the Market Leader Condition. It is easy to observe that in that case, a mirror
situation where $\mathcal{A}_b$ becomes market leader in the Block Traders
Allocation Rule thus producing reserved network effects in favor of it. This
will help define an upper bound on $\phi_a$ for this not to happen.

Formally:

\begin{theorem}[Asymptotic Allocation]\label{asymp-alloc}
  Under the same premises as the Leader Fee Equilibrium proposition,
  \begin{itemize}
    \item $2\sigma(1-\gamma) < \phi_a$  implies $(T_i)_{i\geq 0}$ and
      $(L_i)_{i\geq 0}$ are decreasing and $\lim_{i\to\infty} T_i = 0$.
    \item $\sigma(1-\gamma)T_0^{-1}$ < $\phi_a < 2\sigma(1-\gamma)$ implies
      $(T_i)_{i\geq 0}$ and $(L_i)_{i\geq 0}$ are decreasing and
      $\lim_{i\to\infty} T_i = \sigma(1-\gamma)\phi_a^{-1}$
    \item $\phi_a < \sigma(1-\gamma)T_0^{-1}$ implies $(T_i)_{i\geq 0}$ and $(L_i)_{i\geq 0}$ are
      increasing and $\lim_{i\to\infty} T_i = \min(1, \sigma(1-\gamma)\phi_a^{-1})$
  \end{itemize}
\end{theorem}
\begin{proof}
  First point follows observing that by monotonicity $\exists k>0$ where the
  Market Leader Condition becomes reversed and updating the Traders Allocation
  Rule for $i \geq k$ with $T_i = (1-\sigma)L_{i-1}$.  Second and third point
  follow from assuming Market Leader Condition doesn't change, applying both
  Allocation Rules and noting $T_i \in [0,1]$.
\end{proof}

\section{Stackelberg attack}

Having established the basic equilibrium condition for such AMM forks where the
leader imposes a protocol fee, let's now include the governance of the AMM
(Governance) as a player who can change the AMM protocol fee at any block.

We assume for the sake of the argument that Governance is profit maximizing and
that the protocol fee is entirely directed to itself as a payoff. Note that
Governance would usually be comprised of governance token holders.

As a new starting point, we consider only one AMM without forks yet.

Intuitively, it appears that, as long as LPs' preferences are entirely
represented by their LP Allocation Game utility (notably, if they don't own
governance tokens), then they would be ready to fork the AMM and start using the
feeless fork, any time when it becomes more profitable to do so.

We assume that LPs can commit to new strategies, including deploying smart
contracts that may allocate reserves on their behalf. In line with
\cite{GThBlockchain}, \cite{landis2023stackelberg}, we observe
Stackelberg equilibria emerging from smart contract moves.

We will derive the Grim Forker, a simple smart contract that allows LPs to
commit to fork the AMM and lock their reserves there, whenever the fee is higher
than a defined threshold. This aims at forcing Governance to not move fees higher
than the threshold.

Let's further assume for simplicity that deployment and interaction of
such smart contracts incurs negligible costs to LPs.

\subsection{Governance player}

\begin{definition}[Governance Fee Setting Game]
  Using notation from the Block Allocation Game, a Governed AMM
  $\mathcal{A}_{gov}$'s protocol fee is denoted by a sequence evolving through
  time $(\phi_{gov})_{i\geq0}$.

  The Governance player's action space consists of updating $(\phi_{gov})_i$ at
  any block $i$ with a value in $[0, 1-\gamma)$.

  Governance's payoff at each block $i$ is given by
  $$
  (u_g)_i = (\phi_{gov})_i (V_{gov})_i
  $$

  And its utility is represented by the expected value of payoffs

  $$
  U_g = \sum_{i=i_0}^{\infty}(\phi_{gov})_i (V_{gov})_i\eta^{-i}
  $$

  with $\eta$ a discount factor in $(0,1)$.
\end{definition}

\subsection{Grim Forker contract}

Let's describe Grim Forker, a smart contract which constitutes a Stackelberg
attack in the sense of \cite{landis2023stackelberg}: it enables LPs to commit to
a specific strategy based on some parameters, with the goal of influencing
Governance to reduce the fee as much as possible.

The proposed contract consist of:

\begin{itemize}
  \item a vault (in line with \cite{erc4626}) which acts as a LP itself on
    $\mathcal{A}_{gov}$ and which funds are kept available for withdrawal before
    the fork, but are locked after the fork happens,
  \item a clause to fork the AMM when some key conditions are met, notably
    depending on a maximum fee value $\phi_{threshold}$ which is a parameter of
    the contract.
\end{itemize}

Note that the locking of the funds after the fork aims at producing a large
enough penalty on Governance expected earnings. The duration of this lock can be
modulated as a parameter as well.

Let's describe the forking pseudocode, which is going to be run at each block:

\begin{algorithm}
\caption{Forking Routine}
\begin{algorithmic}
  \IF{$\phi_{gov} > \phi_{threshold}$ \AND $R_{GrimForker} > R_{gov}/2$}
    \STATE Deploy fork of AMM
    \STATE Withdraw reserves from AMM
    \STATE Allocate reserves to fork
    \STATE Prevent withdrawal (for some period)
\ENDIF
\end{algorithmic}
\end{algorithm}

with $R_{GrimForker}$ the amount of reserves that have been allocated to this
contract.

We can define a new, adapted, allocation game to account for the fork.

\begin{definition}[Fork Block Allocation Game]
  Given an AMM $\mathcal{A}_{gov}$ and a its newly forked $\mathcal{A}_{fork}$,
  a Fork Block Allocation Game is constructed by sub-games in the following order: 

\begin{itemize}
  \item Grim Forker LPs Auto-Allocation: execution of the forking routine,
    allocating all $R_{GrimForker}$ to $\mathcal{A}_{fork}$,
  \item Block Traders Allocation Game,
  \item Block LPs Allocation Game.
\end{itemize}
\end{definition}

From this point, we consider the Fork Block Allocation Game in place of the
Block Allocation Game. All results about the Block Allocation Game still apply
as long as the forking clause is not met.

\subsection{Stackelberg equilibrium}

We can now describe the metagame that happens between LPs and Governance where
LPs add the deployment of specific Grim Forker instances to their action space.

\begin{definition}[Grim Forker Game]
  Players are defined by
  \begin{itemize}
    \item Governance, whose utility is $U_g$,
    \item LPs who participate in Grim Forker, or GFLPs.
  \end{itemize}

  The game happens in the following steps:
  \begin{itemize}
    \item GFLPs deploy a version of Grim Forker with some parameters,
    \item Governance adjusts $\phi_{gov}$ accordingly.
  \end{itemize}
\end{definition}

Assuming that more that half of reserves are allocated to Grim Forker, as
Governance aims at maximizing its utility $U_g$, it is presented with two
choices:

\begin{itemize}
  \item prevent Grim Forker from launching the fork while maximizing $U_g$, thus
    setting $\phi_{gov} = \phi_{threshold}$ and gaining $U_g =
    \sum_{i=i_{fork}}^{\infty}\phi_{threshold} V \eta^{-i}$ (which is equal to
    $\phi_{threshold}V(1-\eta)^{-1}$),
  \item let Grim Forker launch the fork if the clause is met and maximize the
    yield of $U_g = \sum_{i=i_{fork}}^{\infty}(\phi_{gov})_i
    (V_{gov})_i\eta^{-i}$ by playing on $\phi_{gov}$.
\end{itemize}

The second option requires more complex modeling but we know that $(V_{gov})_i$
tends towards 0 in this case thanks to Theorem-\ref{asymp-alloc}. It might still
be the rational option if the discount factor is chosen particularly low as part
of Governance's preferences. Such low $\eta$ could reflect a large proportion of
short-term-minded governance participants, who would thus try to influence
Governance in making the second choice.

We can now derive the conditions for the participation of LPs.

\begin{theorem}[Single Grim Forker Participation]\label{single-grim-forker-particip}
  It is ex-interim individual rational for LPs to allocate their reserves to a
  single Grim Forker if
  \begin{itemize}
    \item they don't have conflicting interests with Governance: their entire
      preferences are represented by $u_l$,
    \item they assume that Governance will prefer preventing the fork.
  \end{itemize}
\end{theorem}
\begin{proof}
  First, the equivalent amount of reserves allocated in Grim Forker or directly
  in the AMM yields the same returns. Hence, leveraging the assumption that
  transaction costs are negligible, it is indifferent for LPs to allocate to the
  AMM directly or to Grim Forker as long as the forking is not triggered.

  Second, LPs would prefer allocating their reserves through Grim Forker to keep
  the protocol fee below $\phi_{threshold}$.
\end{proof}

Further refinements to the model could include refining preferences for LPs who
are also part of Governance, adding more costs and risks (e.g., development,
legal), and accounting for delays.

\subsection{Market-sourced equilibrium condition}

Let's restrict our analysis to LPs who have no conflicting interests with
Governance. Further, let's assume that they all prefer preventing the fork from
happening, which is reasonable assuming that most external factors not
considered here (impact on traders' perception and brand recognition) induce a
high cost of forking to LPs.

Under these assumptions, Theorem-\ref{single-grim-forker-particip} applies.
Deploying a Grim Forker contract with a given $\phi_{threshold}$ will attract
LPs who believe that Governance won't let the fork happen.

We suggest that a market mechanism can aggregate LPs' preferences on setting the
$\phi_{threshold}$ that optimizes their aggregate utility under some risk
assumptions. Letting different Grim Forker contracts compete would be suboptimal
because of liquidity fragmentation. This suggests that a bidding mechanism could
be used within the Grim Forker contract, where LPs can set the minimum
$\phi_{threshold}$ they are ready to allocate their reserves to. We leave it to
further work to establish such a mechanism.

As long as bids and reserve allocations are observable, such a mechanism would
produce a signal allowing Governance to adjust the protocol fee to prevent a
fork from happening.

\section{Conclusion}

We have produced a dynamic, block-by-block model for traders' volume allocation
and LPs' reserves allocation among two AMM forks. We have established asymptotic
closed-form expressions for reserves and trade volume allocation ratios.
Assuming reasonable real-world constraints on the protocol fee rate, this result
points towards an indisputable natural monopoly of the market-leading AMM among
its forks, thus removing any benefit of fork-based competition for LPs who are
at risk of losing part or all of their revenue to Governance.

We introduced a Grim Forker contract as a Stackelberg attack by LPs on
Governance and shown that this creates a new equilibrium, giving bargaining
power to LPs. Depending on the participation rate in the contract, this attack
can force Governance into reducing protocol fees, otherwise risking seeing the
original AMM drained of all its reserves.

In turn, Grim Forker contracts, by leaving their participation rate openly
observable, instruct Governance on how to manage their protocol fee optimally.
Conversely, introducing obfuscation techniques as in \cite{darkdaos,
austgen2023dao} would hinder observability and change the game dynamics,
requiring a different modeling treatment.

\subsection{Further work}

While this model is restricted to CPMMs, its utility can be extended to further
types of AMMs. We suggest that, as long as trader and LP utility functions have
similar Taylor expansions, similar results can be derived. Further work could
include a generalization to Constant Function Market Makers
\cite{Angeris2020ImprovedPO}.

An important improvement to make the model more realistic would be to model LPs
as having non-negative governance token holdings. Also, specifying an actual
bidding mechanism to aggregate $\phi_{threshold}$ preferences from LPs would
make this model closer to being implementable.

We observe that the game consisting of, on the one hand, individual LPs deciding
how to allocate their funds between an AMM and its feeless fork and, on the
other hand, Governance deciding on the protocol fee level, is not
\textit{Stackelberg resilient} under \cite{Landis2023WhichGA}. We believe that
the key arguments are contained in the current paper, but further formalization
would be required to fully prove this statement. A valuable follow-up question
to ask is whether there is a protocol-fee mechanism that is Stackelberg
resilient and, if so, how such a mechanism is characterized.

Recent efforts \cite{Sun2023CooperativeAV} have been produced to consider new
equilibria of games involving AI agents leveraging commitment devices, notably
Stackelberg attacks. Analyses of Stackelberg attacks like this paper—and,
ideally, deriving Stackelberg resilient alternatives—appear necessary to produce
future-proof protocol-fee mechanisms that will sustain AI agents as players
while fulfilling mechanism objectives.

\textit{Acknowledgments.} We thank Vaughn McKenzie-Landell for support, Daji
Landis for key feedback on the initial idea, and all reviewers for their
valuable feedback.


\begin{thebibliography}{20}
\providecommand{\natexlab}[1]{#1}
\providecommand{\url}[1]{\texttt{#1}}
\expandafter\ifx\csname urlstyle\endcsname\relax
  \providecommand{\doi}[1]{doi: #1}\else
  \providecommand{\doi}{doi: \begingroup \urlstyle{rm}\Url}\fi

\bibitem[Adams et~al.(2021)Adams, Zinsmeister, Salem, Keefer, and Robinson]{UniswapV3}
H.~Adams, N.~Zinsmeister, M.~Salem, R.~Keefer, and D.~Robinson.
\newblock Uniswap v3 core, 2021.
\newblock URL \url{https://uniswap.org/whitepaper-v3.pdf}.

\bibitem[Angeris and Chitra(2020)]{Angeris2020ImprovedPO}
G.~Angeris and T.~Chitra.
\newblock Improved price oracles: Constant function market makers.
\newblock \emph{Proceedings of the 2nd ACM Conference on Advances in Financial Technologies}, 2020.
\newblock URL \url{https://api.semanticscholar.org/CorpusID:214611887}.

\bibitem[Angeris et~al.(2021)Angeris, Kao, Chiang, Noyes, and Chitra]{angeris2021analysis}
G.~Angeris, H.-T. Kao, R.~Chiang, C.~Noyes, and T.~Chitra.
\newblock An analysis of uniswap markets, 2021.

\bibitem[Austgen et~al.(2023)Austgen, Fábrega, Allen, Babel, Kelkar, and Juels]{austgen2023dao}
J.~Austgen, A.~Fábrega, S.~Allen, K.~Babel, M.~Kelkar, and A.~Juels.
\newblock {DAO} decentralization: Voting-bloc entropy, bribery, and dark {DAOs}, 2023.

\bibitem[Daian et~al.(2018)Daian, Kell, Miers, and Juels]{darkdaos}
P.~Daian, T.~Kell, I.~Miers, and A.~Juels.
\newblock On-chain vote buying and the rise of dark {DAOs}, 2018.
\newblock URL \url{https://hackingdistributed.com/2018/07/02/on-chain-vote-buying/}.

\bibitem[Delaunay et~al.(2023)Delaunay, Frambot, Garchery, and Lesbre]{morphoblue}
M.~G. Delaunay, P.~Frambot, Q.~Garchery, and M.~Lesbre.
\newblock Morpho blue whitepaper, 2023.
\newblock URL \url{https://github.com/morpho-org/morpho-blue/blob/main/morpho-blue-whitepaper.pdf}.

\bibitem[Griffith(2019)]{VirgilGriffith}
V.~Griffith.
\newblock Ethereum is game-changing technology, literally., 2019.
\newblock URL \url{https://medium.com/@virgilgr/ethereum-is-game-changing-technology-literally-d67e01a01cf8}.

\bibitem[Hall-Andersen and Schwartzbach(2021)]{GThBlockchain}
M.~Hall-Andersen and N.~I. Schwartzbach.
\newblock Game theory on the blockchain: A model for games with smart contracts.
\newblock In I.~Caragiannis and K.~A. Hansen, editors, \emph{Algorithmic Game Theory}, pages 156--170, Cham, 2021. Springer International Publishing.
\newblock ISBN 9783-03-08594-7-3.

\bibitem[Huo et~al.(2023)Huo, Klages-Mundt, Minca, Münter, and Wind]{Huo_2023}
L.~Huo, A.~Klages-Mundt, A.~Minca, F.~C. Münter, and M.~R. Wind.
\newblock \emph{Decentralized Governance of Stablecoins with Closed Form Valuation}, page 59–73.
\newblock Springer International Publishing, 2023.
\newblock ISBN 9783-031-18679-0.

\bibitem[Kalai et~al.(2010)Kalai, Kalai, Lehrer, and Samet]{CommitmentFolk}
A.~T. Kalai, E.~Kalai, E.~Lehrer, and D.~Samet.
\newblock A commitment folk theorem.
\newblock \emph{Games Econ. Behav.}, 69:\penalty0 127--137, 2010.
\newblock URL \url{https://api.semanticscholar.org/CorpusID:9444523}.

\bibitem[Landis and Schwartzbach(2023{\natexlab{a}})]{Landis2023WhichGA}
D.~Landis and N.~I. Schwartzbach.
\newblock Which games are unaffected by absolute commitments?
\newblock In \emph{Adaptive Agents and Multi-Agent Systems}, 2023{\natexlab{a}}.
\newblock URL \url{https://api.semanticscholar.org/CorpusID:258557548}.

\bibitem[Landis and Schwartzbach(2023{\natexlab{b}})]{landis2023stackelberg}
D.~Landis and N.~I. Schwartzbach.
\newblock Stackelberg attacks on auctions and blockchain transaction fee mechanisms, 2023{\natexlab{b}}.

\bibitem[Ramirez et~al.(2023)Ramirez, Kolumbus, Nagel, Wolpert, and Jost]{Ramirez2023GameM}
M.~A. Ramirez, Y.~Kolumbus, R.~Nagel, D.~H. Wolpert, and J.~Jost.
\newblock Game manipulators - the strategic implications of binding contracts.
\newblock \emph{ArXiv}, abs/2311.10586, 2023.
\newblock URL \url{https://api.semanticscholar.org/CorpusID:265281305}.

\bibitem[Santoro et~al.(2021)Santoro, t11s, Jadeja, Cañada, and Doggo]{erc4626}
J.~Santoro, t11s, J.~Jadeja, A.~C. Cañada, and S.~Doggo.
\newblock {ERC-4626}: Tokenized vaults, 2021.
\newblock URL \url{https://eips.ethereum.org/EIPS/eip-4626}.

\bibitem[Sun et~al.(2023)Sun, Crapis, Stephenson, Monnot, Thiery, and Passerat-Palmbach]{Sun2023CooperativeAV}
X.~Sun, D.~Crapis, M.~Stephenson, B.~Monnot, T.~Thiery, and J.~Passerat-Palmbach.
\newblock Cooperative {AI} via decentralized commitment devices.
\newblock \emph{ArXiv}, abs/2311.07815, 2023.
\newblock URL \url{https://api.semanticscholar.org/CorpusID:265157983}.

\bibitem[Uniswap({\natexlab{a}})]{gfxuniswap}
Uniswap.
\newblock Uniswap governance forum: Making protocol fees operational, 2023{\natexlab{a}}.
\newblock URL \url{https://gov.uniswap.org/t/making-protocol-fees-operational/21198}.

\bibitem[Uniswap({\natexlab{b}})]{UniswapFeeProposal}
Uniswap.
\newblock Uniswap governance forum: [temperature check] - activate uniswap protocol governance, 2024{\natexlab{b}}.
\newblock URL \url{https://gov.uniswap.org/t/temperature-check-activate-uniswap-protocol-governance/22936}.

\bibitem[Velner et~al.(2017)Velner, Teutsch, and Luu]{cryptoeprint:2017/230}
Y.~Velner, J.~Teutsch, and L.~Luu.
\newblock Smart contracts make bitcoin mining pools vulnerable.
\newblock Cryptology ePrint Archive, Paper 2017/230, 2017.
\newblock URL \url{https://eprint.iacr.org/2017/230}.

\bibitem[von Stackelberg(1934)]{VonStackelberg}
H.~von Stackelberg.
\newblock Marktform und gleichgewicht.
\newblock \emph{Verlag von Julius Springer}, 1934.

\bibitem[Wolpert and Bono(2009)]{GameMining}
D.~Wolpert and J.~Bono.
\newblock Game mining: How to make money from those about to play a game.
\newblock \emph{Advances in Austrian Economics}, 18, 09 2009.
\newblock \doi{10.2139/ssrn.2243767}.

\end{thebibliography}
\end{document}